\documentclass[envcountsect]{llncs}
\usepackage{amsmath}
\usepackage{amssymb}
\usepackage{amsfonts}
\usepackage{array}
\usepackage{float}
\usepackage{enumerate}

\usepackage{ifthen}
\usepackage{xspace}

\usepackage{graphicx}
\usepackage{arydshln}

\usepackage{algorithmic}
\usepackage{algorithm}

\floatstyle{ruled}
\newfloat{algorithm}{thp}{lop}[section]
\floatname{algorithm}{Algorithm}
\newfloat{module}{thp}{lop}[section]
\floatname{module}{Module}

\newtheorem{thm}{Theorem}[section]

\newtheorem{lem}[thm]{Lemma}

\newcommand{\remove}[1]{}

\newenvironment{th-repeat}[1]{\begin{trivlist}
\item[\hspace{\labelsep}{\bf\noindent Theorem~\ref{#1} }]}%
{\end{trivlist}}

\newenvironment{lem-repeat}[1]{\begin{trivlist}
\item[\hspace{\labelsep}{\bf\noindent Lemma~\ref{#1} }]}%
{\end{trivlist}}

\begin{document}

\frontmatter
\pagenumbering{arabic}
\pagestyle{plain}

\title{Optimal Byzantine Resilient Convergence in Asynchronous Robot Networks}

\author{Zohir Bouzid \and Maria Gradinariu Potop-Butucaru \and
S\'{e}bastien Tixeuil}

\institute{Universit\'e Pierre et Marie Curie - Paris 6, LIP6-CNRS 7606, France}

\maketitle

\begin{abstract}
We propose the first deterministic algorithm that tolerates up to $f$ byzantine faults in $3f+1$-sized networks and performs in the asynchronous CORDA model. Our solution matches the previously established lower bound for the semi-synchronous ATOM model on the number of tolerated Byzantine robots. Our algorithm works under bounded scheduling assumptions for oblivious robots moving in a uni-dimensional space. 
\end{abstract}

\section{Introduction} 

\emph{Convergence} is a fundamental agreement primitive in robot networks and is used in the implementation of a broad class of services (\emph{e.g.} the construction of common coordinate systems or specific geometrical patterns). 
Given a set of oblivious robots with arbitrary initial locations and no agreement on a global coordinate system, \emph{convergence} requires that all robots asymptotically approach the same, but unknown beforehand, location. Convergence is hard to achieve in asynchronous systems, when robots obtain information only \emph{via} visual sensors, since they are unable to distinguish between a moving or a stationary robot. The problem becomes even harder when some robots are Byzantine (\emph{i.e.} those robots can exhibit arbitrary behavior). In that case correct robots are required to converge independently of the behavior of the faulty ones. 

Robots operate in \emph{cycles} that comprise \emph{look}, \emph{compute}, and \emph{move} phases. The look phase consists in taking a snapshot of the other robots positions using its visibility sensors. In the compute phase a robot computes a target destination based on the previous observation. The move phase simply consists in moving toward the computed destination using motion actuators. 
The robots that we consider have weak capacities: they are \emph{anonymous} (they execute the same protocol and have no mean to distinguish themselves from the others), \emph{oblivious} (they have no memory that is persistent between two cycles), and have no compass whatsoever (they are unable to agree on a common direction or orientation).

In order to capture the essence of distributed coordination in robot networks, two main computational models are proposed in the literature: the ATOM~\cite{SY99} and CORDA~\cite{Pre05} models. The main difference between the two models comes from the granularity of the execution of the \emph{look-compute-move} cycle. For the ATOM model, the cycle is atomic while in CORDA the atomicity concerns only the elementary phases in the cycle. That is, in the ATOM model, robots executing concurrently remain in the same phase (they are synchronous or semi-synchronous) while in CORDA they are asynchronous (\emph{e.g.} a robot may execute the look phase while another robot performs its move phase).

\paragraph{Related works}

Since the pioneering work of Suzuki and Yamashita~\cite{SY99}, gathering\footnote{Gathering requires  robots to actually \emph{reach} a single point within finite time regardless of their initial positions.} and convergence have been addressed in \emph{fault-free} systems for a broad class of settings. Prencipe~\cite{Pre05} studied the problem of gathering in both ATOM and CORDA models, and showed that the problem is intractable without additional assumptions such as being able to detect the multiplicity of a location (\emph{i.e.}, knowing if there is more than one robot in a given location). Flocchini \emph{et~al.}~\cite{FPS+05} proposed a gathering solution for oblivious robots with limited visibility in the CORDA model, where robots share the knowledge of a common direction given by a compass. The subsequent work by Souissi \emph{et~al.}~\cite{TR:SDY05} consider a system in which compasses are not necessarily consistent initially. In~\cite{ando1999dmp} the authors address convergence with limited visibility in fault-free environments. Convergence with inaccurate sensors and movements is addressed in \cite{CP06}. Recently,  in \cite{sirocco09} the authors study the same problem under a uniform sensing error model. Ando \emph{et~al.}~\cite{ando1999dmp} propose a gathering algorithm for the ATOM model with limited visibility. 

The case of \emph{fault-prone} robot networks was recently tackled by several academic studies. The faults that have been investigated fall in two categories: \emph{crash} faults (\emph{i.e.} a faulty robots stops executing its cycle forever) and \emph{Byzantine} faults (\emph{i.e.} a faulty robot may exhibit arbitrary behavior and movement). Of course, the Byzantine fault model encompasses the crash fault model, and is thus harder to address.
\emph{Deterministic} fault-tolerant gathering is addressed in\cite{agmon2004ftg} where the authors study a gathering protocol that tolerates one crash, and an algorithm for the ATOM model with fully synchronous scheduling that tolerates up to $f$ byzantine faults, when the number of robots is (strictly) greater than $3f$. In \cite{defago3274fta} the authors study the feasibility of \emph{probabilistic} gathering in crash-prone and Byzantine-prone environments.
\emph{Deterministic} fault-tolerant convergence was first addressed in~\cite{cohen2004rcv,cohen2005cpg}, where algorithms based on convergence to the center of gravity of the system are presented. Those algorithms work in CORDA model and tolerate up to $f$ ($n>f$) crash faults, where $n$ is the number of robots in the system. Most realted to this paper is~\cite{BGT09c}, where the authors studied convergence in byzantine-prone environments when robots move in a uni-dimensional space. In more details, \cite{BGT09c} showed that convergence is impossible if robots are not endowed with strong multiplicity detectors which are able to detect the exact number of robots that may simultaneously share the same location. The same paper defines the class of \emph{cautious} algorithms which guarantee that correct robots always move inside the range of positions held by correct robots, and proved that any cautious convergence algorithm that can tolerate $f$ Byzantine robots requires the presence of at least $2f+1$ robots in fully-synchronous ATOM networks and $3f+1$ robots in semi-synchronous ATOM networks. The ower bound for the ATOM model naturally extends to the CORDA model, yet the protocol proposed in~\cite{BGT09c} for the asynchronous CORDA model requires at least $4f+1$ robots.

\begin{table}
\centering
\begin{tabular}{|c|l|l|l|}\hline
\textbf{Reference} & \textbf{Model} & \textbf{Faults} & \textbf{Bounds} \\\hline
\cite{CP06} & ATOM & inaccurante sensors& - \\
& &movements and calc.&\\\hline
\cite{agmon2004ftg} & ATOM & crash & $f=1$ \\
& Fully-sync. ATOM& \textbf{Byzantine} & $n>3f$ \\\hline
\cite{cohen2004rcv} & ATOM & crash & $n>f$ \\\hline 
\cite{cohen2005cpg} & \textbf{CORDA} & crash & $n>f$ \\\hline
\cite{BGT09c} & Fully-Sync ATOM & \textbf{Byzantine} & $n>2f$ \\
& ATOM & \textbf{Byzantine} & $n>3f$\\
& \textbf{CORDA} & \textbf{Byzantine} & $n>4f$\\\hline
\textbf{This paper} & \textbf{CORDA} & \textbf{Byzantine} & $n>3f$\\\hline
\end{tabular}
\caption{Crash and byzantine resilience bounds for deterministic gathering and convergence}
\label{tab:results}
\end{table}

Table~\ref{tab:results} summarizes the results related to crash and byzantine resilience of gathering and convergence deterministic protocols that are known in robot netwoks. The bold values denote the least specialized (and more difficult) hypotheses.

\paragraph{Our contributions}
In this paper we consider the class of cautious algorithms, which guarantees that correct robots always move inside the range of positions held by correct robots. In this class, we propose an optimal (with respect to the number of Byzantine robots) Byzantine resilient solution for convergence when robots execute their actions in the CORDA model. That is, our solution tolerates $f$ byzantine robots in $3f+1$-sized networks, which matches the lower bound presented in~\cite{BGT09c} for the class of cautious algorithms.   

\paragraph{Outline}
The remaining of the paper is organized as follows: Section~\ref{sec:model} presents our model and robot network assumptions. Section~\ref{sec:problem} presents the formal specification of the convergence problem and recalls and necessary and sufficient conditions to achieve convergence in Byzantine prone systems, Section~\ref{sec:algorithm} describes our protocol and its complexity, while concluding remarks are presented in Section~\ref{sec:conclusion}.

\section{Model}
\label{sec:model}

Most of the notions presented in this section are borrowed from\cite{SY99,Pre01,agmon2004ftg}.  We consider a network that consists of a finite set of robots arbitrarily deployed in a uni-dimensional space. The robots are devices with sensing, computing and moving capabilities. They can observe (sense) the positions of other robots in the space and based on these observations, they perform some local computations that can drive them to other locations. 

In the context of this paper, the robots are \emph{anonymous}, in the sense that they can not be distinguished using their appearance, and they do not have any kind of identifiers that can be used during the computation. In addition, there is no direct mean of communication between them. Hence, the only way for robots to acquire information is by observing their positions. Robots have \emph{unlimited visibility}, \emph{i.e.} they are able to sense the entire set of robots. Robots are also equipped with a strong multiplicity sensor referred to as \emph{multiples detector} and denoted hereafter by $\mathcal{M}$. This sensor provides robots with the ability to detect the exact number of robots that may simultaneously occupy the same location\footnote{In \cite{BGT09c}, it is proved that $\mathcal{M}$ is necessary to deterministically solve the convergence problem in a uni-dimensional space even in the presence of a single Byzantine robot.}. We assume that the robots cannot remember any previous observation nor computation performed in any previous step. Such robots are said to be \emph{oblivious} (or \emph{memoryless}). 

A \emph{protocol} is a collection of $n$ \emph{programs}, one operating on each robot. The program of a robot consists in executing {\em Look\mbox{-}Compute\mbox{-}Move cycles} infinitely many times. That is, the robot first observes its environment (Look phase). An observation returns a snapshot of the positions of all robots within the visibility range. In our case, this observation returns a snapshot (also called \emph{configuration} hereafter) of the positions of \emph{all} robots denoted with $P(t)= \{P_1(t), ... , P_n(t)\}$. The positions of correct robots are referred as $U(t)=\{U_1(t), ... , U_m(t)\}$ where $m$ denotes the number of correct robots. Note that $U(t) \subseteq P(t)$. The observed positions are \emph{relative} to the observing robot, that is, they use the coordinate system of the observing robot. We denote by $P^i(t)= \{P_1^i(t), ... , P_n^i(t)\}$ the configuration $P(t)$ given in terms of the coordinate system of robot $i$ ($U^i(t)$ is defined similarily).
Based on its observation, a robot then decides --- according to its program --- to move or stay idle (Compute phase). When an robot decides a move, it moves to its destination during the Move phase.  An \emph{execution} $e=(c_0, \ldots, c_t, \ldots)$ of the system is an infinite sequence of configurations, where $c_0$ is the initial configuration\footnote{Unless stated otherwise, we make no specific assumption regarding the respective positions of robots in initial configurations.} of the system, and every transition $c_i \rightarrow c_{i+1}$ is associated to the execution of a subset of the previously defined actions. 

A \emph{scheduler} is a predicate on computations, that is, a scheduler define a set of \emph{admissible} computations, such that every computation in this set satisfies the scheduler predicate. A \emph{scheduler} can be seen as an entity that is external to the system and selects robots for execution. As more power is given to the scheduler for robot scheduling, more different executions are possible and more difficult it becomes to design robot algorithms. In the remaining of the paper, we consider that the scheduler is \emph{$k$-bounded} if, between any two activations of a particular robot, any other robot can be activated at most $k$ times\footnote{Note that~\cite{BGT09c} proves the impossibility result with $n=3f$ robots using a $2$-bounded scheduler.}. 

We now review the main differences between the ATOM~\cite{SY99} and CORDA~\cite{Pre01} models. In the ATOM model, whenever a robot is activated by the scheduler, it performs a \emph{full} computation cycle. Thus, the execution of the system can be viewed as an infinite sequence of rounds. In a round one or more robots are activated by the scheduler and perform a computation cycle. The \emph{fully-synchronous ATOM} model refers to the fact that the scheduler activates all robots in each round, while the \emph{semi-synchronous ATOM} model enables the scheduler to activate only a subset of the robots.
In the CORDA model, robots may be interrupted by the scheduler after performing only a portion of a computation cycle. In particular, actions (look, compute, move) of different robots may be interleaved. For example, a robot $a$ may perform a look phase, then a robot $b$ performs a look-compute-move complete action, then $a$ computes and moves based on its previous observation (that does not correspond to the current configuration anymore). As a result, the set of executions that are possible in the CORDA model are a strict superset of those that are possible in the ATOM model. So, an impossibility result that holds in the ATOM model also holds in the CORDA model, while an algorithm that performs in the CORDA model is also correct in the ATOM model. Note that the converse is not necessarily true.
 
The faults we address in this paper are \emph{Byzantine} faults. A byzantine (or malicious) robot may behave in arbitrary and unforeseeable way. In each cycle, the scheduler determines the course of action of faulty robots and the distance to which each non-faulty robot will move in this cycle. However, a robot is guaranteed to move a distance of at least $\delta$ towards its destination before it can be stopped by the scheduler.

Our convergence algorithm performs operations on multisets. A multiset or a bag $S$ is a generalization of a set where an element can have more than one occurence. The number of occurences of an element \emph{a} is referred as its \emph{multiplicity} and is denoted by $mul(a)$. The total number of elements of a multiset, including their repeated occurences, is referred as the \emph{cardinality} and is denoted by $|S|$. $\min(S)$(resp. $\max(S)$) is the smallest (resp. largest) element of $S$. If $S$ is nonempty, $range(S)$  denotes the set $[\min(S), \max(S)]$ and $diam(S)$ (diameter of $S$) denotes $\max(S) - \min(S)$.

\section{The Byzantine Convergence Problem}
\label{sec:problem}

Given an initial configuration of $n$ autonomous mobile robots ($m$ of which are correct such that $m \geq n-f$), the \emph{point convergence problem} requires that all correct robots asymptotically approach the exact same, but unknown beforehand, location. In other words, for every $\epsilon > 0$, there is a time $t_\epsilon$ from which all correct robots are within distance of at most $\epsilon$ of each other.

\begin{definition}[Byzantine Convergence]
\label{def:byz-convergence}
A system of oblivious robots satisfies the Byzantine convergence specification if and only if $\forall \epsilon > 0, \exists t_\epsilon$ such that $\forall t > t_\epsilon$, $\forall$ i,j $\leq m, \mathit{distance}(U_i(t), U_j(t)) < \epsilon$, where $U_i(t)$ and $U_j(t)$ are the positions of some \emph{correct} robots $i$ and $j$ at time $t$, and where $\mathit{distance}(a,b)$ denote the Euclidian distance between two positions. 
\end{definition}

Definition~\ref{def:byz-convergence} requires the convergence property only from the \emph{correct} robots. Note that it is impossible to obtain the convergence for all robots since Byzantine robots may exhibit arbitrary behavior and never join the position of correct robots.

In the following we recall the necessary conditions to achieve convergence in systems prone to Byzantine failures. We first focus on the definition of \emph{shrinking} algorithms (algorithms that eventually decrease the range between any two correct robots). In \cite{BGT09c} is proved that this condition is necessary but not sufficient for convergence even in fault-free environments. We then recall the defintion of \emph{cautious} algorithms (algorithms that ensure that the position of correct robots always remains inside the range of the correct robots). This condition combined with the previous one is sufficient to reach convergence in fault-free systems~\cite{BGT09c}. 

By definition, convergence aims at asymptotically decreasing the range of possible positions for the correct robots. The shrinking property captures this property. An algorithm is shrinking if there exists a constant factor $\alpha \in (0,1)$ such that starting in any configuration the range of correct robots eventually decreases by a multiplicative $\alpha$ factor. 

\begin{definition}[Shrinking Algorithm] 
\label{def:shrink}
An algorithm is \emph{shrinking} if and only if $\exists \alpha \in (0,1)$ such that $\forall t, \exists t^\prime > t$, such that $diam(U(t^\prime)) < \alpha*diam(U(t))$, where $U(t)$ is the multiset of positions of correct robots.
\end{definition}

A natural way to solve convergence is to never let the algorithm increase the diameter of correct robot positions. In this case the algorithm is called \emph{cautious}. This notion was first introduced in~\cite{dolev1986raa}. A cautious algorithm is particularly appealing in the context of Byzantine failures since it always instructs a correct robot to move inside the range of the positions held by the correct robots regardless of the locations of Byzantine ones. The following definition introduced first in \cite{BGT09c} customizes the definition of cautious algorithm proposed in \cite{dolev1986raa} to robot networks.

\begin{definition}[Cautious Algorithm]
\label{def:caut}
Let $D_i(t)$ the last destination calculated by the robot $i$ before time $t$ and let $U^i(t)$ the positions of the correct robots as seen by robot $i$ before time $t$. \footnote{If the last calculation was executed at time $t^\prime \leq t$ then $D_i(t) = D_i(t^\prime)$.} An algorithm is \emph{cautious} if it meets the following conditions:
\begin{itemize}
\item \textbf{cautiousness: } $\forall t,~D_i(t) \in range(U^i(t))$ for each robot $i$.
\item \textbf{non-triviality: } $\forall t$, if $diameter(U(t)) \neq 0$ then $\exists t^\prime>t$ and a robot $i$ such that $D_i(t^\prime)\neq U_i(t^\prime)$ (at least one correct robot changes its position).
\end{itemize}
\end{definition}

The following theorem will be further used in order to prove the correctness of our convergence algorithm.

\begin{thm}\cite{BGT09c}
\label{th:cands}
Any algorithm that is both cautious and shrinking solves the convergence problem in fault-free robot networks.
\end{thm}

\section{Deterministic Asynchronous Convergence}
\label{sec:algorithm}

In this section we propose a deterministic convergence algorithm and prove its correctness in the CORDA model under a k-bounded scheduler.
The idea of Algorithm \ref{alg-convergence-3f} is as follows: each robot computes the median of the positions of the robots seen in its last Look phase ignoring the $f$ largest positions if they are larger than his own position and the $f$ smallest positions if they are smaller than his own position.

Algorithm \ref{alg-convergence-3f} uses two functions, $trim_f^i()$ and $center()$.
The choice of the  function $trim_f^i()$ makes the difference between this algorithm and that of~\cite{BGT09c}. Indeed, in~\cite{BGT09c} the trimming function removes the $f$ largest and the $f$ smallest values from the multiset given in parameter. That is, the returned multiset does not depend on the position of the calling robot. In Algorithm~\ref{alg-convergence-3f}, $trim_f^i()$ removes among the $f$ largest positions \emph{only} those that are greater than the position of the calling robot $i$. Similarly, it removes among the $f$ smallest positions only those that are smallest than the position of the calling robot. 

Formally, let $minindex_i$ be the index of the minimum position between $P_i(t)$ and $P_{f+1}(t)$ (if $P_i(t)<P_{f+1}(t)$ then $minindex_i$ is equal to $i$, otherwise it is equal to $f+1$). Similarily, let $maxindex_i$ be the index of the maximum position between $P_i(t)$ and $P_{n-f}(t)$ (if $P_i(t)> P_{n-f}(t)$ then $maxindex_i$ is equal to $i$, otherwise it is equal to $n-f$). $trim_f^i(P(t))$ is the multiset consisting of positions $\{P_{minindex_i}(t), P_{minindex_i+1}(t), ..., P_{maxindex_i}(t)\}$.
$center()$ returns the median point of the input range. The two functions are illustrated in Figure \ref{fig-illustration-algo}) . 

\begin{figure}[htbp]
\begin{center}
\includegraphics[height=4cm,width=6cm]{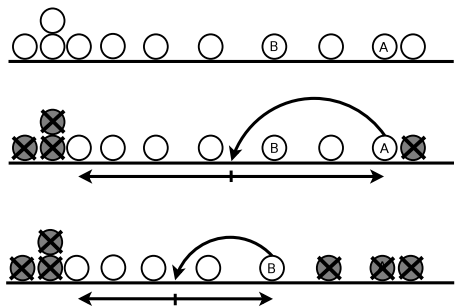}
\caption{Illustration of functions $trim_f$ and $center$ for robots $A$ and $B$.}
\label{fig-illustration-algo}
\end{center}
\end{figure}

\begin{algorithm}
\caption{Byzantine Tolerant Convergence}          
\label{alg-convergence-3f}                  
\begin{algorithmic}
\STATE \textbf{Functions}:\\

\STATE - $trim_f^i(P(t))$: removes up to $f$ largest positions that are larger than $P_i(t)$ and up to $f$ smallest positions that are smaller than $P_i(t)$ from the multiset $P(t)$ given in parameter.
\STATE - $center$: returns the point that is in the middle of the range of points given in parameter.
\STATE

\STATE   \textbf{Actions}:\\

\STATE move towards $center(trim_f^i(P(t)))$
\end{algorithmic}
\end{algorithm}

In the following we prove the correctness of Algorithm \ref{alg-convergence-3f} in the CORDA model under a $k$-bounded scheduler. In order to show that Algorithm \ref{alg-convergence-3f} converges, we prove first that it is cautious then we prove that it satisfies the specification of a shrinking algorithm. Convergence then follows from Theorem \ref{th:cands}.

\subsection{Algorithm \ref{alg-convergence-3f} is cautious}

In this section we prove that Algorithm \ref{alg-convergence-3f}  is a cautious algorithm (see Definition \ref{def:caut}) for $n>3f$. The following lemma states that the range of the trimmed multiset $trim_f^i(P(t))$ is contained in the range of correct positions.

\begin{lem}
\label{trim_sub_range}
Let $i$ be a correct robot executing Algorithm \ref{alg-convergence-3f}, 
it holds that 
\[range(trim_f^i(P(t))) \subseteq range(U(t))
\]
\end{lem}

\begin{proof}
We prove that for any correct robot, $i$, the following conditions hold: 
\begin{enumerate}
\item $\forall t~ min(trim_f^i(P(t))) \in range(U(t)).$
\item $\forall t~ max(trim_f^i(P(t))) \in range(U(t)).$
\end{enumerate}

\begin{enumerate}
\item 
By definition, $min(trim_f^i(P(t)))=min\{P_i(t), P_{f+1}(t)\}$. 
Hence proving Property (1) reduces to proving 
$P_i(t) \in range(U(t))$ and $P_{f+1}(t) \in range(U(t))$.
Similarly, proving property (2) reduces to proving $P_i(t) \in range(U(t))$ and $P_{n-f}(t) \in range(U(t))$

\begin{itemize}
\item $P_i(t) \in range(U(t))$ directly follows from the assumption that  robot $i$ is correct.
\item $P_{f+1}(t) \in range(U(t))$.
Suppose the contrary: there exists some time instant $t$ such that $P_{f+1}(t) \notin range(U(t))$ 
and prove that this leads to a contradiction. If $P_{f+1}(t) \notin range(U(t))$ 
then either $P_{f+1}(t) < U_1(t)$ or $P_{f+1}(t) > U_m(t)$.

\begin{itemize}
\item If $P_{f+1}(t) < U_1(t)$ then there are at least $f+1$ positions $\{P_1(t),$ $P_2(t),$ $\ldots,$ $P_f(t),$ $P_{f+1}(t)\}$ 
that are smaller than $U_1(t)$ which is the \emph{first} correct position in the network at time $t$. This means that 
there would be at least $f+1$ byzantine robots in the system. But this contradicts the assumptions 
that at most $f$ byzantine robots are present in the system.

\item If $P_{f+1}(t) > U_m(t)$ then since $n>3f$ there are more than $f$ positions $\{P_{f}(t), P_{f+1}(t), ...,P_n(t)\}$ 
that are greater than $U_m(t)$, which is the \emph{last} correct position in the system at time $t$. This also leads to a contradiction.

\end{itemize}
\end{itemize}

\item In the following we prove that $P_i(t) \in range(U(t))$ and $P_{n-f}(t) \in range(U(t))$.

\begin{itemize}

\item $P_i(t) \in range(U(t))$since robot $i$ is correct.

\item $P_{n-f}(t) \in range(U(t))$.
Assume the contrary: there exists some time instant $t$ such that $P_{n-f}(t) \notin range(U(t))$ 
and prove that this leads to a contradiction. If $P_{n-f}(t) \notin range(U(t))$ then either $P_{n-f}(t) < U_1(t)$ or $P_{n-f}(t) > U_m(t)$.

\begin{itemize}
\item If $P_{n-f}(t) > U_m(t)$ then there are at least $f+1$ positions $\{P_{n-f}(t),$  $\ldots,$ $P_{n-1}(t),$ $P_{n}(t)\}$ that are greater than $U_m(t)$ which is the \emph{last} correct position in the network at time $t$. 
It follows that there would be at least $f+1$ byzantine robots in the system which contradicts the hypothesis on the maximal number of faulty robots in the system.  

\item If $P_{n-f}(t) < U_1(t)$ then since $n>3f$ there are more than $f$ positions $\{P_1(t), P_2(t), ...,P_{n-f}(t)\}$ 
that are smaller than $U_1(t)$, which is the \emph{first} correct
position in the system at time $t$. 
This also leads to a contradiction.

\end{itemize}
\end{itemize}

\end{enumerate}
\end{proof}

A direct consequence of the above property is that correct robots always compute a destination within the range of positions held by correct robots, whatever the behavior of Byzantine ones. Thus, the diameter of positions held by correct robots never increases. Consequently, the algorithm is cautious. The formal proof is proposed in the following lemma.

\begin{lem}
\label{cautious-FS}
Algorithm \ref{alg-convergence-3f} is cautious for $n>3f$.
\end{lem}

\begin{proof}
According to Lemma \ref{trim_sub_range}, $range(trim_f^i(P(t))) \subseteq
range(U(t))$ for each correct robot $i$, 
thus $center(trim_f^i(P(t))) \in range(U(t))$. 
It follows that all destinations computed by correct robots are
located 
inside $range(U(t))$ which proves the lemma.
\end{proof}

\subsection{Algorithm \ref{alg-convergence-3f} is shrinking}

In this section we prove that Algorithm~\ref{alg-convergence-3f} is a shrinking algorithm (see Definition~\ref{def:shrink})
The following lemma states that a robot can not compute a destination that is far from its current position by more than half the diameter of correct positions.
More specifically, a robot located on one end of the network can not move to the other end in a single movement. 

Interestingly, the property of lemma \ref{lem-halfdiameter} is guaranteed even though robots are not able to figure out the range of correct positions nor to compute the corresponding diameter. 
The bound on the movements of robots is achieved by taking into account the position of the calling robot when computing the trimming function. 
It is important to note that if all robots compute their destinations using the same trimming function irrespective of the position of the calling robot, convergence requires the presence of more than $4f$ robots to tolerate the presence of up to $f$ Byzantine robots~\cite{BGT09c}.

\begin{lem}
\label{lem-halfdiameter}
$\forall t, \forall i$, correct robot, if $i$ computes its destination point at time $t$, then at $t$, $distance(U_i^i(t), D_i(t)) \leq diameter(U^i(t)) / 2$
\end{lem}

\begin{proof}
Suppose the contrary: $distance(U_i^i(t), D_i(t)) > diameter(U^i(t)) /
2$ for some robot 
$i$ at time $t$. Assume without loss of generality that $U_i^i(t) <
D_i(t)$ (the other case is symmetric).
So $U_i^i(t) < D_i(t)+ diameter(U^i(t))/2$ and we prove that this is impossible.

Recall that $D_i(t)$ is the center of $trim_f^i(P(t))$ which implies
that if there exists 
some point $p \in trim_f^i(P(t))$ with $p<D_i(t)$ such that
$distance(p,D_i(t)) > diameter(U(t))/2$, 
then there must exist another point $q \in trim_f^i(P(t))$ with
$q>D_i(t)$ such that 
$distance(D_i(t),q) > diameter(U(t))/2$.
This implies that $distance(p,q)>range(U^i(t))$. Since both $p$
and $q$ belong to 
$trim_f^i(P(t))$ then $diameter(trim_f^i(P(t)) \geq distance(p,q) > diameter(U^i(t))$.
This contradicts lemma \ref{trim_sub_range} which states that $range(trim_f^i(P(t)) \subseteq range(U(t))$.
\end{proof}

The following lemmas describe some important properties on the destination points computed by correct robots which will be used in proving the shrinkingness of Algorithm~\ref{alg-convergence-3f}. These properties are verified whatever the positions of Byzantine robots are, and thus they capture the limits of the influence of Byzantine robots on the actions undertaken by correct robots.

The next lemma shows that the correct positions $\{U_{f+1}(t), ... ,U_{m-f}(t) \}$ are always included in 
the trimmed range (the output range of the function $trim_f^i$) regardless of the positions of Byzantine robots.

\begin{lem}
\label{inside_range}
It holds that $range(trim_f(U(t))) \subseteq  range(trim_f(P(t)))$.
\end{lem}

\begin{proof}
We prove that: 
\begin{enumerate}
\item $\forall t~ U_{f+1}(t) \in range(trim_f(P(t))).$
\item $\forall t~ U_{m-f}(t) \in range(trim_f(P(t))).$
\end{enumerate}

\begin{enumerate}
\item 
Suppose that $U_{f+1}(t) \notin range(trim_f(P(t)))$. 
Then either 
\[U_{f+1}(t) < \min(trim_f(P(t)))
\] or 
\[U_{f+1}(t) > \max(trim_f(P(t)))
\]

\begin{itemize}

\item If $U_{f+1}(t) < \min(trim_f(P(t)))$ then there are at least $f+1$ positions $\{U_1(t), ..., U_{f+1}(t)\}$ which are smaller than $\min(trim_f(P(t)))$. This contradicts the definition of $trim_f(P(t))$ (at most $f$ among the smallest elements of $P(t)$ are removed). 

\item If $U_{f+1}(t) > \max(trim_f(P(t)))$ and since $|U(t)| > 2f$
(because $n>3f$), then there are also at least $f+1$ 
positions in $U(t)$ greater than $\max(trim_f(P(t)))$, which also leads to a contradiction.

\end{itemize}

\item The property is symmetric to the precedent one.
\end{enumerate}
\end{proof}

Let $D(t)$ be the set of destinations computed with 
Algorithm~\ref{alg-convergence-3f} in systems with $n>3f$, and let $UD(t)$ be the union of $U(t)$ and $D(t)$.
If a robot $i$ executed its last Look action at time $t^\prime \leq t$, then $UD^i(t)= UD(t^\prime)$.
The following lemma proves that the destination computed by each correct robot $i$ is always 
within the range $[(min(UD^i(t)) + U_{m-f}^i(t))/2 , (U_{f+1}^i(t) +
max(UD^i(t)))/2]$ independently of the positions of Byzantine robots.

\begin{lem}
\label{destination_mid}
The following properties hold: 
\item $\forall i$, each destination point calculated by a correct robot
$i$ at time $t$ is (1) smaller than $(U_{f+1}^i(t) + max(UD^i(t)))/2$ and (2) greater than 
$(min(UD^i(t)) + U_{m-f}^i(t))/2$.
\end{lem}

\begin{proof}
Let $d_1$ be the distance between $U_{f+1}^i(t)$ and $max(UD^i(t))$.
\begin{enumerate}
\item 
We suppose the contrary: there exists some calculated destination point $D_i$ by some correct robot 
$i$ at time $t$ such that \[D_i > (U_{f+1}^i(t) + max(UD^i(t)))/2
\] and we prove that this leads to a contradiction.
\[D_i > (U_{f+1}^i(t) + max(UD^i(t)))/2
\] implies that $U_{f+1}^i(t) <
D_i-d_1/2$. And by Lemma~\ref{inside_range}, 
$U_{f+1}^i$(t) is inside 
\[range(trim_f^i(P^i(t)))
\] which means that
there is a position inside 
$range(trim_f^i(P^i(t)))$ which is smaller than $D_i-d_1/2$.
Hence there must exists a position inside \[range(trim_f^i(P^i(t)))\]
say $p$, 
such that $p>D_i+d_1/2$ because $D_i$ is the center of $trim_f^i(P^i(t))$. 
$U_{f+1}^i(t) < D_i-d_1/2$ and $p>D_i+d_1/2$ implies that 
$distance(U_{f+1}(t), p)> distance (U_{f+1}(t), max(UD^i(t))$ which in
turn implies that $p>max(UD^i(t))$. But $p \in
range(trim_f^i(P^i(t)))$, it 
follows that \[max(trim_f^i(P^i(t))) > max(UD^i(t))\] which contradicts Lemma \ref{trim_sub_range} and thereby proves our lemma.

\item Symmetric to the precedent property.
\end{enumerate}
\end{proof}

\begin{lem}
\label{uf-less-st}
Let $S(t)$ be a multiset of $f+1$ arbitrary elements of U(t). The following properties hold: 
(1) $\forall t,~ U_{f+1}(t) \leq \max(S(t))$
and (2) $\forall t,~ U_{m-f}(t) \geq \min(S(t))$
\end{lem}

\begin{proof}
\begin{enumerate}
\item Assume the contrary: $U_{f+1}(t) > \max(S(t))$. This means that
$U_{f+1}(t)$ is 
strictly greater than at least $f+1$ elements of U(t), which leads to
a contradiction 
($U_{f+1}(t)$ is by definition the $(f+1)$-th correct position in $U(t)$).
\item The property is symmetric to the precedent.
\end{enumerate}
\end{proof}

The next lemma generalizes and extends the properties of Lemmas
\ref{inside_range} and \ref{destination_mid} 
(proven for a fixed time instant) to a time interval. It describes
bounds on the destination points 
computed by correct robots during a time interval $[t_1, t_2]$. It
states that if there is a subset of 
$f+1$ robots whose positions are less than $S_{max}$ during  $[t1,t2]$, then all destinations computed 
during $[t1, t2]$ by all correct robots in the network are necessarily smaller than $(S_{max}+Max(UD(t_1)))/2$.

\begin{lem}\label{limit-destinations}
Let a time $t_2>t_1$ and let $S(t)$ be a multiset of $f+1$ arbitrary elements in U(t).
If $\forall p \in S(t)$ and  $\forall t \in [t_1, t_2]~ p \leq
S_{max}$  then all 
calculated destination points at time interval $[t_1, t_2]$ are smaller than $(S_{max}+Max(UD(t_1)))/2$.
\end{lem}

\begin{proof}
By definition of $S_{max}$ we have that $\forall t \in [t_1, t_2], max(S(t)) \leq S_{max}$.
According to Lemma~\ref{uf-less-st}, $\forall t \in [t_1, t_2]~  U_{f+1}(t) \leq max(S(t))$. 
So $\forall t \in [t_1, t_2],$ $U_{f+1}(t)$ $\leq$ $S_{max}$.

By Lemma~\ref{destination_mid}, 
each calculated destination point by each correct robot $i$ at time interval $[t_1, t_2]$ is smaller than 
$(U_{f+1}^i(t)+max(UD(t)))/2$, so because $U_{f+1}(t) \leq S_{max}$
these 
destinations points are also smaller than $(S_{max}+max(UD(t)))/2$ .
Since the algorithm is cautious, $\forall i, \forall t  \in [t_1,
t_2]~ max(UD(t)) \leq max(UD(t_1)) $ and the lemma follows.

\end{proof}

The next Lemma states that if some calculated destination point is in
the neighborhood of one end of the network, then a majority of $m-f$
correct robots are necessarily located in the neighborhood of this end.

\begin{lem}
\label{lem-majority}
If some correct robot $i$ executes its Look action at time $t$ and
then compute (in the Compute action which immediatly follows) a
destination $D_i$ such that $D_i<min(UD(t))+b$ (with $b$ any distance
smaller than $diameter(UD(t))/2$), then at t, there are at least $m-f$
correct robots whose positions are (strictly) smaller than $min(UD(t))+2b$.
\end{lem}

\begin{proof}
We prove first that at $t$, $max(trim_f^i(P(t))) <= min(UD(t)) + 2b$.
According to Lemma~\ref{trim_sub_range}, $min(trim_f^i(P(t))) \geq min(UD(t))$. 
And we have by hypothesis that $D_i < min(UD(t))+b$. This gives us $D_i < min(trim_f^i(P(t)))+b$. 
But $D_i$ is the center of $trim_f^i(P(t))$ which means that
$distance(D_i, min(trim_f^i(P(t))))$ 
must be equal to $distance(D_i, max(trim_f^i(P(t))))$. 
Thus, $max(trim_f^i(P(t)))) < D_i + b$. And since by hypothesis $D_i <
min(UD(t))+b$, 
we have \[max(trim_f^i(P(t)))) < min(UD(t)) + 2b\]

which means that at $t$ there are
at 
most $f$ correct positions greater than $min(UD(t)) + 2b$, and by
definition 
no correct position is smaller than $min(UD(t))$. 
It follows that at $t$, the range $[min(UD(t)),$ $min(UD(t))+2b)$ contains at least $m-f$ correct positions.
\end{proof}

We are now ready to give the proof of shrinkingness of our algorithm in the CORDA model. The general idea of the proof is to show that the destination points computed by correct robots are located either around the middle of the range of correct positions or/and in the neighborhood of only one end of this range.

If all computed destinations are located around the middle of the
range of correct robots 
then the diameter of this range decreases and the algorithm is
shrinking. 
Otherwise, if some computed destinations are located in the
neighborhood of one end of the range, 
it is shown that there is a time at which no correct robot will be in
the neighborhood 
of the other end of the range, which leads again to a decrease in the
range 
of correct positions and shows that the algorithm is shrinking.

In this section we address the correctness of Algorithm \ref{alg-convergence-3f} in the CORDA model under a $k$-bounded scheduler. Our proof is constructed on top of the auxiliary lemmas proposed in the previous section.

\begin{lem}
\label{shrinking_corda}
Algorithm \ref{alg-convergence-3f} is shrinking in the CORDA model with $n>3f$ under a k-bounded scheduler.
\end{lem}

\begin{proof}
Let $U(t_0)=\{U_1(t_0), ..., U_m(t_0)\}$ be the configuration of correct robots at
initial time $t_0$ and $D(t_0)=\{D_1(t_0), ..., D_m(t_0)\}$ the
multiset of their calculated 
destination points at the same time $t_0$ and $UD(t_0)$ is the union
of $U(t_0)$ and $D(t_0)$. 
Let $t_1$ be the first time at which all correct robots have been
activated and executed their 
Compute cycle at least once since $t_0$ ($U(t_1)$ and $D(t_1)$ are the corresponding multisets of positions and destinations).
Assume that robots are ordered 
from left to right and define $d_0$ and $d_1$ as their diameters at $t_0$ and $t_1$ respectively. Since the model is asynchronous, the diameter calculation takes into account both the positions and the destinations of robots. So $d_0=diameter (UD(t_0)$ and $d_1=diameter (UD(t_1))$. Let $b$ be any distance that is smaller than $d_0/4$, for example take $b=d_0/10$.

We consider the actions of correct robots after $t_1$ and we separate the analysis into two cases:

\begin{itemize}
\item \emph{Case A}: All calculated destinations by all correct robots after
$t_1$ are inside $[min(UD(t_0))+b, max(UD(t_0))-b]$. 
So when all correct robots are activated at least once, their diameter
decreases by at least 
$\min\{2\delta, 2b=d_0/5\}$. Thus by setting $\alpha_1=\max\{1-2\delta/d_0, 4/5\}$, the algorithm is shrinking.
\linebreak

\item \emph{Case B}: Let $t_2 > t_1$ be the first time when a robot, say $i$,
execute a 
Look action such that the Compute action that follows compute a destination point, say $D_i$, that is outside $[min(UD(t_0))+b, max(UD(t_0))-b]$. 
This implies that either ($D_i < min(UD(t_0))+b$) or ($D_i > max(UD(t_0))-b$). 
Since the two cases are symmetric, we consider only the former which
implies according to 
Lemma \ref{lem-majority} that the range $[min(U D(t_0)) , min(UD(t_0))+2b]$ must contain at least $m-f$ correct positions.

If some robots among these $m-f$ robots are executing a Move action,
their destination points have necessarily 
been calculated after $t_0$ (since at $t_1$ each robot has been
activated at least once). 
And we have by lemma \ref{lem-halfdiameter} that the distance between
each robot and its destination can not exceed half the diameter, so we
conclude that at $t_2$ 
the destination points of these $m-f$ robots are all inside $[min(UD(t_0)) , min(UD(t_0))+b+d_0/2]$.

Let $S(t_2)$ be a submultiset of $UD(t_2)$ containing the positions
and destinations of $f+1$ arbitrary robots among these $m-f$ whose
positions and destinations 
are inside $[min(UD(t_0)) , min(UD(t_0))+b+d_0/2]$. So
$\max(S(t_2))<=min(UD(t_0))+b+d_0/2$. 
And since we choosed $b<d_0/4$, we have  $\max(S(t_2)) < max(UD(t_0))-3d_0/4$.
Let $t_3\geq t_2$ be the first time each correct robot in the system has been activated at least once since $t_2$. 
We prove in the following that at $t_3$, $\max(S(t_3)) < max(UD(t_0))-3d_0/2^{k(f+1)+2}$.

To this end we show that the activation of a single robot of
$S(t)$ can not reduce the distance between the 
upper bound of $max(S)$ and $max(UD(t_0))$ by more than half its
precedent value, and since 
the scheduler is k-bounded, we can guarantee that this distance at
$t_3$ 
is at least equal to $3d_0/2^{k(f+1)+2}$.

According to Lemma~\ref{destination_mid}, if some robot $i$ calculates
its destination $D_i$ at time $t \in [t_2, t_3]$, $D_i \leq
(U_{f+1}(t) + max(UD(t)))/2$. 
But $U_{f+1}(t) \leq max(S(t))$ by Lemma \ref{uf-less-st} and $max(UD(t)) <= max(UD(t_0))$ due to cautiousness. 
This gives us $D_i \leq (max(S(t) + max(UD(t_0)))/2$.
Therefore, an activation of a single robot in $S(t)$ to execute its
Compute cycle can reduce the 
distance between $Max(UD(t_0))$ and $max(S(t))$ by at most half its precedant value.
 
So at $t_3$, after a maximum of $k$ activations of each robot in
$S(t)$, we have 
$max(S(t_3)) <= Max(UD(t_0))-3d_0/2^{k(f+1)+2}$, and by Lemma~\ref{limit-destinations}, all calculated
destinations by all correct robots between $t_2$ and $t_3$ are less than or equal to $Max(UD(t_0))-3d_0/2^{k(f+1)+3}$. 

Since robots are guaranteed to move toward their destinations by at
least a distance $\delta$ before they can be stopped by the scheduler,
after $t_3$, no robot will be located beyond $Max(UD(t_0))-min\{\delta, 3d_0/2^{k(f+1)+3}\}$.
Hence by setting $\alpha=\max\{\alpha_1, 1-\delta/d_0,1-3/2^{k(f+1)+3}\}$ the lemma follows.

\end{itemize}

\end{proof}

The convergence proof of Algorithm \ref{alg-convergence-3f} directly follows from 
Lemma \ref{shrinking_corda} and Lemma \ref{cautious-FS}.

\begin{thm}
Algorithm \ref{alg-convergence-3f} solves the Byzantine convergence problem in the CORDA model for
$n>3f$ under a $k$-bounded scheduler.
\end{thm}

\section{Conclusions and discussions}
\label{sec:conclusion}
In this paper we consider networks of oblivious robots with arbitrary initial locations and no agreement on a global coordinate system. Robots obtain system related information only 
\emph{via} visual sensors and some of them are Byzantine (\emph{i.e.} they can exhibit arbitrary behavior). 
In this weak scenario, we studied the \emph{convergence} problem that requires that all
robots to asymptotically approach the exact same, but unknown beforehand,
location. We focused on the class of cautious algorithms, which guarantees that correct robots 
always move inside the range of positions held by correct robots. 
In this class we proposed an optimal byzantine resilient solution 
when robots execute their actions asynchronously as defined in the CORDA model. That is, 
our solution tolerates $f$ byzantine robots in $3f+1$-sized networks, which matches previously established lower bound.
  
Two immediate open problems are raised by our work:
\begin{itemize}
\item Our algorithm is proved correct under bounded scheduling assumption. We conjecture that this hypothesis is necessary for achieving convergence in the class of cautious algorithms.
\item The study of asynchronous byzantine-resilient convergence in a multi-di\-men\-sio\-nal space is still open. 
\end{itemize} 

\bibliographystyle{plain}
\bibliography{convergence}
\end{document}